\documentclass[11pt,a4]{article}
\usepackage{amsfonts}
\usepackage{graphics,amssymb,amsthm,amsmath,epsf}
\usepackage{graphicx}
\usepackage{subfigure}

\usepackage[round]{natbib}

\newtheorem{algorithm}{Algorithm}[section]
\newtheorem{theorem}{Theorem}[section]

\numberwithin{equation}{section}
\newtheorem{remark}{Remark}[section]

\begin{document}

\date{}
\title{Comparing the Shape Parameters of Two Weibull Distributions Using Records: A Generalized Inference }

\author{H. Zakerzadeh , A. A. Jafari\thanks{Corresponding: aajafari@yazd.ac.ir} \\
{\small Department of Statistics, Yazd University, Yazd, Iran}}

\date{}
\maketitle
\begin{abstract}
 The Weibull distribution is a very applicable model for the lifetime data. For inference about two Weibull distributions using records, the shape parameters of the distributions are usually considered equal. However, there is not an appropriate method for comparing the shape parameters in the literature.  Therefore, comparing the shape parameters of two Weibull distributions is very important. In this paper, we propose a method for constructing confidence interval and testing hypotheses about the ratio and difference of shape   parameters using the concept of the generalized {\it p}-value and the generalized confidence interval. Simulation studies showed that our method is satisfactory. In the end, a real example is proposed to illustrate this method.

\end{abstract}



\section{Introduction}

\cite{chandler-52}
introduced the concept of record value and studied some its properties. Record data arise in a wide variety of practical situations; for example industrial stress testing, meteorological analysis, sporting and athletic events, and mining surveys. Properties of record data have been extensively studied in the literature.
 \cite{ahsanullah-95}
and
\cite{ar-ba-na-98}
are two good references about records and their properties.

Let $X_1,X_2,\dots $ be a sequence of independent and identically distributed  continuous random variables having the same distribution as the (population) random variable $X$. The random variable $X_k$ is an upper record value if it is greater than all preceding values $X_1,X_2,\dots ,X_{k-1}$. The sequence of record time $\{T_n, n\ge 0\}$ is defined as follows:
\[T_0=1\ \ \ {\rm with\ probability}\ 1,\ \ \ {\rm and}\ \ \ \ T_n=\min\left\{i:\ X_i>X_{T_{n-1}}\right\} \ \ \ \ \ n\ge 1.\]
The sequence of upper record values is defined by $\{R_n=X_{T_n}, n=0,1,\dots \}$. By definition, $X_1$ is taken as the first upper record value. In the same way, an analogous definition can be provided for lower record values.

Suppose that we observe the first $n+1$ upper record values $R_0, R_1,\dots ,R_n$ from the cumulative distribution function (cdf)  $F_{\boldsymbol\theta }(x)$ and the probability density function (pdf)  $f_{\boldsymbol \theta }(x)$, where $\boldsymbol \theta$ is a vector of parameters. Then, the joint distribution of the first $n+1$ record values
\citep[for more details see][page 10]{ar-ba-na-98}
 is given by
\[f_{\boldsymbol R}\left({\boldsymbol r}\right)=f_{\boldsymbol\theta }\left(r_n\right)\prod^{n-1}_{i=0}{\frac{f_{\boldsymbol\theta }\left(r_i\right)}{1-F_{\boldsymbol\theta }\left(r_i\right)}},\ \ \ \ \ \ \ \ \ r_0<r_1<\dots <r_n,\]
where ${\boldsymbol R}=(R_0,R_1,\dots,R_n)$ and ${\boldsymbol r}=(r_0,r_1,\dots ,r_n)$.

Some researches have considered inference about the Weibull distribution based on records:
\cite{ho-pa-94} discussed the maximum likelihood estimates (MLE's) for both scale and shape parameters of a Weibull distribution.
\cite{chan-98}
and \cite{su-ba-99}
presented some inferential methods for   the location-scale families  of the distributions.
  Exact confidence intervals and exact joint confidence regions for the parameters  of a Weibull distribution are derived by \cite{chen-04}.
\cite{wu-ts-06}
proposed  a computational approach for inference about the shape parameter.
\cite{so-ab-su-06}
obtained the MLE's for the parameters of a Weibull distribution  and developed a Bayesian analysis using record values.
  Exact joint confidence regions for the parameters are also derived by \cite{as-ab-11}, meanwhile
\cite{te-gu-12} proposed a confidence interval for $n$th upper/lower record value.
\cite{te-na-13}
derived exact expressions for constructing bias corrected MLE's.
When the shape parameters of two Weibull distributions are equal, the stress-strength parameter of these distributions has a closed form. For such a case,
\cite{baklizi-12}
proposed some methods for estimating and constructing confidence interval for the  parameter of stress-strength reliability based on record values.
 However, it seems that there is no method for inference about the shape parameters of two Weibull distributions.

In this paper, we have considered constructing confidence interval and testing the hypothesis about the ratio (and difference) of two shape parameters. This is an extension of the method  proposed by
\cite{wu-ts-06}
for the shape parameter of one Weibull distribution. For inference, we have applied the concepts of generalized confidence interval and generalized {\it p}-value introduced by
 \cite{ts-we-89}
and
\cite{weerahandi-93},
respectively.
These approaches have been used successfully to address several complex problems
\citep[see][]{weerahandi-95}
such as inference about the mean of a Weibull distribution
\citep{kr-li-xi-09},
the stress-strength reliability involving two independent Weibull distributions
\citep{kr-li-10},
the stress-strength reliability in two-parameter exponential distribution \citep{baklizi-13},
inference on common mean of several normal populations \citep{kr-lu-03},
inference on common mean of several log-normal populations \citep{be-ja-06-ge}
and comparing two generalized variances of multivariate distributions \citep{jafari-12}.

The rest of the present article is organized as follows: In Section \ref{sec.GPV}, we briefly review the concepts of generalized confidence interval and generalized {\it p}-value. A method for inference about the ratio and difference of two shape parameters is proposed in Section \ref{sec.inf}.
In Section \ref{sec.sim}, we investigate the performance of the proposed approach using a simulation study. An  illustrative example is proposed in Section \ref{sec.ex}.

\section{ Generalized {\it p}-value and generalized confidence interval}
\label{sec.GPV}
Let ${\boldsymbol X}$ be a random variable whose distribution depends on a vector parameters $\boldsymbol\theta=(\tau,\boldsymbol\lambda)$, where $\tau$ is
a scale parameter of interest and $\boldsymbol\lambda$ is a vector of nuisance parameters. Let ${\boldsymbol x}$ denotes the observed value of ${\boldsymbol X}$. A generalized pivotal quantity for $\tau $ is a random quantity denoted by $T({\boldsymbol X};{\boldsymbol x};\tau)$ and satisfies the following conditions:

\noindent(i) The distribution of $T({\boldsymbol X};{\boldsymbol x};\tau)$ is free of any unknown parameters.

\noindent(ii) The value of $T({\boldsymbol X};{\boldsymbol x};\tau)$ at ${\boldsymbol X} = {\boldsymbol x}$, i.e., $T({\boldsymbol x};{\boldsymbol x};\tau )$ is free of the nuisance parameter $\boldsymbol\lambda $.  In most cases, $T\left({\boldsymbol x};{\boldsymbol x};\tau \right)=\tau$.

Appropriate percentiles of $T({\boldsymbol X};{\boldsymbol x};\tau )$ form a confidence interval for $\tau$. Specifically, if $T_{\delta }$  denotes the
100$\delta $ percentage point of $T({\boldsymbol X};{\boldsymbol x};\tau)$, then $(T_{\gamma /2}, T_{1-\gamma /2})$ is a $100(1-\gamma)\% $ generalized confidence interval for $\tau$. Because, for a given ${\boldsymbol x}$, the distribution of $T({\boldsymbol X};{\boldsymbol x};\tau )$ does not depend on any unknown parameters, its percentiles can be found.

In the above setup, suppose that we are interested in testing the hypotheses
\begin{equation}\label{eq.H0}
H_0:\tau \le {\tau }_0\ \ \ \ \ \ \ \ vs.\ \ \ \ \ \ \ H_1:\tau >{\tau }_0,
\end{equation}
for a specified known  ${\tau}_0$. The generalized test variable, denoted by $T^*({\boldsymbol X};{\boldsymbol x};\tau )$, is defined as follows:

\noindent(i) The value of $T^*(\boldsymbol X;\boldsymbol x;\tau )$ at $\boldsymbol X= \boldsymbol x$ is free of any unknown parameters.

\noindent(ii) The distribution of $T^*({\boldsymbol X};{\boldsymbol x};\tau )$ is stochastically monotone (i.e., stochastically increasing or stochastically decreasing) in $\tau$ for any fixed ${\boldsymbol x}$ and  $\boldsymbol\lambda$.

\noindent(iii) The distribution of $T^*({\boldsymbol X};{\boldsymbol x};\tau )$ is free of any unknown parameters.

Let $t^*=T^*\left({\boldsymbol x};{\boldsymbol x};{\tau }_0\right)$ denotes the observed value of $T^*({\boldsymbol X};{\boldsymbol x};\tau )$ at $\left({\boldsymbol X};\tau \right)=({\boldsymbol x};{\tau }_0)$. When the above three conditions in (i)-(iii) hold, the generalized \textit{p}-value for testing the hypotheses in \eqref{eq.H0} is defined as
\begin{equation}\label{eq.pvG}
p=P\left(T^*\left(\boldsymbol X;\boldsymbol x;\tau_0\right)\le t^*\right),
\end{equation}
if $T^*({\boldsymbol X};{\boldsymbol x};\tau )$ is stochastically decreasing in $\tau $. In many situations, $T^*\left({\boldsymbol X};{\boldsymbol x};\tau \right)=$ \break $T({\boldsymbol X};{\boldsymbol x};\tau)-\tau $, where $T({\boldsymbol X};{\boldsymbol x};\tau)$ is a generalized pivotal variable. The test based on the generalized \textit{p}-value rejects $H_0$ when the generalized \textit{p}-value is smaller than a given level $\gamma $. However, the size and the power function of such a test may depend on the nuisance parameters.

For more details on generalized \textit{p}-values and generalized confidence intervals, we refer readers to
\cite{weerahandi-95}.

\section{ Inference about the parameters}
\label{sec.inf}
The Weibull distribution with parameters $\alpha $ and $\beta $ has the pdf
\[F\left(x\right)=1-e^{-{\left(\frac{x}{\alpha}\right)}^\beta},\ \ \ x>0,\ \ \ \alpha>0,\ \ \ \beta>0,\]
and the cdf
\[f\left(x\right)=\frac{\beta }{\alpha^\beta}x^{\beta-1}e^{-{\left(\frac{x}{\alpha}\right)}^\beta},\ \ \  x>0.\]

This distribution is a generalization of the exponential distribution and the Rayleigh distribution. Also, $Y={\log (X)}$ has extreme value (Gumbel) distribution with parameters $b=\frac{1}{\beta}$ and $a=\log(\alpha)$, when $X$ has a Weibull distribution with parameters $\alpha $ and $\beta $.
It is a well-known distribution that is widely used for lifetime models while having numerous varieties of shapes and being very flexible such that it has both increasing and decreasing failure rates. Based on this, the Weibull distribution is used for many applications such as hydrology, reliability engineering, weather forecasting and insurance.

Suppose $\boldsymbol R_i=(R_{i0},R_{i1},\dots ,{R_{in}}_i)$, $i=1,2$ are the set of records corresponding to an independent and identically sequence
 of a Weibull distribution with parameters ${\alpha }_i$ and ${\beta}_i$.
In this section, we consider constructing confidence interval for the ratio of the shape parameters, $\pi=\frac{\beta_1}{\beta_2}$, and
 testing the one-sided hypotheses
\begin{equation}\label{eq.H01}
H_0:\pi \le \pi_0\ \ \ \ \ \ vs.\ \ \ \ \ \ H_1:\pi >\pi_0,
\end{equation}
and the two-sided hypotheses
\begin{equation}\label{eq.H02}
H_0:\pi =\pi_0\ \ \ \ \ \ vs.\ \ \ \ \ \ H_1:\pi \ne \pi_0,
\end{equation}
where $\pi_0$ is a specified value.

 The cdf of the record values, $\boldsymbol R_i$,  can be written as
\begin{eqnarray*}
f_{\boldsymbol R_i}(\boldsymbol r_i)
=\frac{{\beta }^{n_i+1}_i}{{\alpha }^{{\beta }_i\left(n_i+1\right)}_i}e^{-{\left(\frac{r_{in_i}}{\alpha_i}\right)}^{\beta_i}}\prod^{n_i}_{j=0}{r^{\beta_i-1}_{ij}}, \ \ \ \ \ 0<r_{i0}<r_{i1}<\dots<{r_{in}}_i,
\end{eqnarray*}
where $\boldsymbol r_i=(r_{i0},r_{i1},\dots ,{r_{in}}_i)$. Therefore, $\left(R_{in_i},\prod^{n_i}_{j=0}{R_{ij}}\right)$ is a sufficient statistic for $\left(\alpha _i,\beta_i\right)$.
 In addition, the MLE's of the parameters $\alpha_i$ and $\beta_i$ are \citep[see][]{so-ab-su-06}
\begin{equation}\label{eq.MLE}
{\hat{\beta }}_i=\frac{n_i+1}{\sum^n_{i=0}{{\log  \left(\frac{R_{in_i}}{R_{ij}}\right)}}},\ \ \ \ \ \ \ {\hat{\alpha }}_i=\frac{R_{in_i}}{{\left(n_i+1\right)}^{\frac{1}{\hat{\beta_i}}}}.
\end{equation}

Based on  the equality of shape parameters of  two Weibull distributions, i.e. ${\beta }_1={\beta }_2=\beta $,  the joint density function of these record values can be written as
\begin{align*}
f_{\boldsymbol R_1, \boldsymbol R_2}(\boldsymbol r_1,\boldsymbol r_2)
=\frac{\beta^{n_1+n_2+2}}{\alpha^{\beta \left(n_1+1\right)}_1{\alpha }^{\beta \left(n_2+1\right)}_2} \ e^{-{\left(\frac{r_{1n_1}}{{\alpha }_1}\right)}^{\beta }-{\left(\frac{r_{2n_2}}{{\alpha }_2}\right)}^{\beta }}\prod^{n_1}_{j=0}{r^{\beta-1}_{1j}}\prod^{n_2}_{h=0}{r^{\beta-1}_{2h}}.
\end{align*}
Therefore, $\left(R_{1n_1},R_{2n_2},\prod^{n_1}_{j=0}{R_{1j}}\prod^{n_2}_{h=0}{R_{2h}}\right)$ is a sufficient statistic for $\left(\alpha_1,\alpha_2,\beta \right)$, and the MLE's of the parameters ${\alpha }_1$, ${\alpha }_2$ and $\beta$ are \citep[see][]{baklizi-12}
\begin{equation}
\hat{\beta }=\frac{n_1+n_2+2}{\sum^{n_1}_{j=0}{\log \left(\frac{R_{1n_1}}{R_{1j}}\right)}+\sum^{n_2}_{j=0}{{\log\left(\frac{R_{2n_2}}{R_{2j}}\right)}}},\ \ \ \  \hat{\alpha}_i=\frac{R_{in_i}}{(n_i+1)^{\frac{1}{\hat{\beta }}}} \ \ \ i=1,2.
\end{equation}

 \cite{wu-ts-06}
has proposed an approach for inference about the shape parameter of a Weibull distribution.  We will use this method for inference about $\pi$, and propose a generalized confidence interval for this parameter as well as a generalized test variable for testing the  hypotheses in \eqref{eq.H01} and \eqref{eq.H02}.

Let
\[W_i\left({\beta }_i\right)=\frac{\sum^{n_i}_{j=0}{R^{{\beta }_i}_{ij}}}{(n_i+1){\left(\prod^{n_i}_{j=0}{R_{ij}}\right)}^{\frac{{\beta }_i}{n_i+1}}},\ \ \ \ \ \ i=1,2.\]
\cite{wu-ts-06} show that $W_i({\beta }_i)$ is an increasing function with respect to $\beta_i$.  Also, the distribution of $W_i\left({\beta }_i\right)$ does not depend on parameters ${\alpha }_i$ and ${\beta }_i$. In fact, $W_i(\beta _i)$ is distributed as
\[W^*_i=\frac{\sum^{n_i}_{j=0}{R^*_{ij}}}{(n_i+1){\left(\prod^{n_i}_{j=0}{R^*_{ij}}\right)}^{\frac{1}{n_i+1}}},\]
where $R^*_{i0},R^*_{i1},\dots ,R^*_{in_i}$ is the record values from the standard exponential distribution. However, the exact distribution of $W_i(\beta_i)$ is very complicated, and its percentiles are  obtained  using the Monte Carlo simulation.

Let
\begin{equation}\label{eq.g}
{{\rm g}}_i\left({\beta }_i\right)=\frac{\sum^{n_i}_{j=0}{r^{{\beta }_i}_{ij}}}{(n_i+1){\left(\prod^{n_i}_{j=0}{r_{ij}}\right)}^{\frac{{\beta }_i}{n_i+1}}}-\frac{\sum^{n_i}_{j=0}{R^*_{ij}}}{(n_i+1){\left(\prod^{n_i}_{j=0}{R^*_{ij}}\right)}^{\frac{1}{n_i+1}}},
\end{equation}
where $r_{ij}$ is the observed value of $R_{ij}$, $i=1,2$, $j=0,1,\dots ,n_i$, and $R^*_{i0},R^*_{i1},\dots ,R^*_{in_i}$ are the record values from the standard exponential distribution.

\begin{theorem}\label{thm.ti}
Let $T_i$ be the solution of the following equations with respect to ${\beta }_i$:
\begin{equation}\label{eq.T}
{{\rm g}}_i\left({\beta }_i\right)=0,\ \ i=1,2.
\end{equation}
Then

\noindent i. $T_i$ is unique.

\noindent ii. $T_i$ is a generalized pivotal variable for $\beta_i$.
\end{theorem}

\begin{proof}
\noindent i. Consider $\bar{R}^*_i=\frac{1}{n_i+1}\sum^{n_i}_{j=0}R^*_{ij}$ and $\bar G^*_i=\left(\prod^{n_i}_{j=0}{R^*_{ij}}\right)^{\frac{1}{n_i+1}}$, $i=1,2$ are the arithmetic mean and geometric mean of $R^*_{i0}, R^*_{i1},\dots R^*_{in_i}$. It is well-known $\bar G^*_i < \bar{R}^*_i$. Therefore,
$$
\lim_{\beta_i\rightarrow 0}{{\rm g}}_i\left(\beta_i\right)=1-\frac{\bar G^*_i}{\bar{R}^*_i}<0, \ \ \ \ \ \  \lim_{\beta_i\rightarrow \infty}{{\rm g}}_i\left(\beta_i\right)=\infty.
$$
Also,
${\rm g}_i(\beta_i)$ is an increasing function with respect to $\beta_i$ \citep[for more deltalis see][]{wu-ts-06}. So,  $T_i$ is unique.

\noindent ii. It is obvious using the substitution approach described by   \cite{weerahandi-04}, page 24.
\end{proof}

 Based on the Theorem \ref{thm.ti}, it can be  understood that i) the observed value of $T_i$ is $\beta_i$ and does not depend on the nuisance parameter, $\alpha_i$, and ii) the distribution of $T_i$ does not depend on any parameter.  Now define
\begin{equation}\label{eq.geG}
G=\frac{T_1}{T_2}.
\end{equation}
Therefore, $G$ is a generalized pivotal variable for $\pi$ and can be used for constructing confidence interval for this parameter. A generalized test variable  can also be defined as
\[G^*=G-\pi.\]
The cdf of $G^*$ is
$F_{G^*}(x)=F_G(x+\pi)$, where $F_G(.)$ is the cdf of the generalized pivotal variable $G$ in \eqref{eq.geG} and does not depend on any parameter. Therefore, $F_{G^*}(x)$ is an increasing function with respect to $\pi$,  and   $G^*$ is stochastically decreasing with respect to $\pi$, and   the generalized {\it p}-values for testing the one-sided hypothesis in \eqref{eq.H01} and  \eqref{eq.H02} are
\begin{eqnarray}\label{eq.pvG1}
&&p=P(G^*<0|\pi_0)=P(G<\pi_0),\\
\label{eq.pvG2}
&&p=2\min  \left\{P(G<{\pi }_0), P(G>{\pi }_0)\right\},
\end{eqnarray}
respectively. This generalized confidence interval and the generalized {\it p}-values can be obtained using Monte Carlo simulation. To do this,  an algorithm is given in Section \ref{sec.sim}.

\begin{remark}
A generalized pivotal approach can also be defined for difference between two shape parameters, ${\beta }_1-{\beta }_2$ as
$H=T_1-T_2$.
\end{remark}

\section{ Simulation study}
\label{sec.sim}
A simulation study is performed to assess the accuracy of the proposed generalized procedure. We evaluated the coverage probability and the expected length of the 95\% generalized confidence 
about $\pi =\beta_1/\beta_2$. To do this,
without loss of generality, we set $\alpha_1=\alpha_2=1$ and
 use Monte Carlo simulation by the following algorithm:

\begin{algorithm}\label{alg.1}
For given $\beta_1$ and $\beta_2$,

\noindent 1. Two sets of records, $r_{i0},\dots ,r_{in_i}$, ($i=1, 2$) were generated from the Weibull distributions.

\noindent 2. Generate the record values $R^*_{i0},\dots ,R^*_{in_i}$ from the standard exponential distribution.

\noindent 3. Write  the equations ${\rm g}_i(\beta_i)$, $i=1, 2$ in \eqref{eq.g} and obtain $T_i$ by solving the equations in \eqref{eq.T}.

\noindent 4. Calculate $G=T_1/T_2$.

\noindent 5. Repeat Steps 2-4,  $M=10,000$ times and obtain the values $G_1,\dots .,G_M$.

\noindent 6. Sort the values of $G_l$, denoted by $G_{(1)},\dots,G_{(M)}$.
The $100(1-\gamma)\%$ generalized confidence for $\pi $ is $\left[G_{(\gamma M/2)},G_{((1-\gamma /2)M)}\right]$.

\noindent 7. Set $D_l=1$ if $G_{(\gamma M/2)}<\frac{\beta_1}{\beta_2}<G_{((1-\gamma /2)M)}$, otherwise $D_l=0$.

\noindent 8. Repeat Steps 1-7, $N=10000$ times. Then coverage probability is $\frac{1}{N}\sum^N_{l=1}{D_l}$.

\end{algorithm}

 For $\beta_2=2$, and some selected values for $\beta_1$, $n_1$, and $n_2$, the coverage probabilities and the expected lengths of the  generalized confidence interval, with 10000 repetition, are given in Table \ref{tab.sim1}. Empirically, we can conclude that

\noindent i. The coverage probability of our method is close to the nominal confidence coefficient.

\noindent ii. For fixed $n_1$ and $n_2$, the expected length of the method is increasing in the parameter shape, $\beta_1$.

\noindent iii. For fixed $\beta_1$ and for  fixed $n_1$, the expected length of the method is decreasing in  $n_2$.

\noindent iv. For fixed $\beta_1$ and for  fixed $n_2$, the expected length of the method is decreasing in  $n_1$.

\begin{table}[h]
\begin{center}
\caption{ Empirical coverage probabilities and expected lengths of the 95\% generalized confidence interval.}
\label{tab.sim1}
\begin{tabular}{|l|l|c|c|c|c|c|c|c|} \hline
\multicolumn{2}{|c|}{} & \multicolumn{7}{|c|}{${\beta }_1$} \\ \hline
 & $n_1,n_2$ & 0.5 & 1.0 & 1.2 & 1.5 & 2.0 & 3.0 & 5.0 \\ \hline
Empirical & 3,3 & 0.946 & 0.952 & 0.953 & 0.946 & 0.948 & 0.949 & 0.952 \\
Coverage & 3,7 & 0.952 & 0.948 & 0.951 & 0.953 & 0.954 & 0.947 & 0.948 \\
 & 3,14 & 0.950 & 0.952 & 0.951 & 0.948 & 0.953 & 0.946 & 0.944 \\
 & 7,3 & 0.956 & 0.948 & 0.952 & 0.946 & 0.953 & 0.954 & 0.951 \\
 & 7,7 & 0.953 & 0.951 & 0.948 & 0.950 & 0.953 & 0.947 & 0.945 \\
 & 7,14 & 0.952 & 0.952 & 0.953 & 0.947 & 0.946 & 0.953 & 0.954 \\
 & 14,3 & 0.945 & 0.952 & 0.954 & 0.950 & 0.947 & 0.949 & 0.952 \\
 & 14,7 & 0.948 & 0.953 & 0.945 & 0.949 & 0.952 & 0.954 & 0.944 \\
 & 14,14 & 0.951 & 0.948 & 0.950 & 0.950 & 0.949 & 0.946 & 0.953 \\ \hline
 \hline
Expected & 3,3 & 2.567 & 4.372 & 5.440 & 6.306 & 9.438 & 11.664 & 24.562 \\
Length & 3,7 & 1.266 & 2.681 & 3.503 & 4.179 & 5.244 & 9.541 & 13.034 \\
 & 3,14 & 0.987 & 2.463 & 2.221 & 3.443 & 3.851 & 6.891 & 11.660 \\
 & 7,3 & 1.567 & 2.786 & 4.116 & 5.200 & 5.852 & 10.089 & 15.026 \\
 & 7,7 & 0.908 & 1.680 & 2.157 & 2.705 & 3.390 & 5.169 & 8.050 \\
 & 7,14 & 0.641 & 1.306 & 1.648 & 2.025 & 2.904 & 3.958 & 6.460 \\
 & 14,3 & 1.418 & 2.608 & 3.522 & 4.425 & 5.310 & 7.814 & 13.786 \\
 & 14,7 & 0.711 & 1.521 & 1.698 & 2.125 & 3.062 & 4.191 & 7.348 \\
 & 14,14 & 0.523 & 1.077 & 1.207 & 1.691 & 2.040 & 2.913 & 5.195 \\ \hline
\end{tabular}
\end{center}
\end{table}

\section{ An illustrative example}
\label{sec.ex}
In this section, we have consider a real data, due to \cite{nelson-82},
concerning the data on time to breakdown of an insulating fluid between electrodes at two voltages of 34 and 36 kV (minutes).
 This data set  is also given by \citet[][page 3]{lawless-03}.
The times to breakdown at voltages of 34 kV and 36 kV are given bellow;
\bigskip
\begin{center}
\begin{tabular}{crrrrrrrrrr}
Voltage of 34 kV: & 0.96  & 4.15 & 0.19 & 0.78 & 8.01  & 31.75 & 7.35   & 6.50 & 8.27 & 33.91 \\
                  & 32.52 & 3.16 & 4.85 & 2.78 & 4.67  &  1.31 &12.06  &36.71 &72.89 &  \\
Voltage of 36 kV: & 1.97  & 0.59 & 2.58 & 1.69 & 2.71  & 25.50 & 0.35   & 0.99 & 3.99 & 3.67 \\
                  & 2.07  & 0.96 & 5.35 & 2.90 &13.77  &  &  &  &  &  \\
\end{tabular}
\end{center}
\smallskip

%
%

\noindent Therefore, the  upper record values at voltage of 34 kV are
0.96, 4.15, 8.01, 31.75, 33.91, 36.71, 72.89, and at voltage of 36 kV are 1.97, 2.58, 2.71, 25.50.

A model suggested by engineering considerations is that, for a fixed voltage level, the time to breakdown has a Weibull distribution \citep{so-ab-su-06}. Based on \eqref{eq.MLE}, the MLE's of the parameters are
$\hat{\beta}_1=0.5990$,  $\hat{\beta}_2=0.5639$, $\hat{\alpha}_1=2.8303,$ $\hat{\alpha}_2=2.1822$, and their standard errors using the Hessian matrix are
$s.e.(\hat{\beta}_1)=0.2264$,  $s.e.(\hat{\beta}_2)=0.2820$, $s.e.(\hat{\alpha}_1)=3.9072,$ $s.e.(\hat{\alpha}_2)=3.3074$.

The \%95 generalized confidence interval for $\pi=\beta_1/\beta_2$ is $(0.2550, 4.9537)$. At the same time,  \%95  generalized confidence interval for $\beta_1-\beta_2$ is
$(-0.7849, 0.7283)$. Also, we consider testing the equality of shape parameters of two Weibull distributions, i.e. $H_0:\beta_1=\beta_2$ vs. $H_1:\beta_1\neq\beta_2$. Using the
algorithm \ref{alg.1} with $\pi_0=1$, the generalized {\it p}-value for testing this hypotheses is 0.9830. So, it can be concluded that the shape parameters of two Weibull distributions are equal, i.e. $\beta_1=\beta_2=\beta $ at level 0.05. In this case, the MLE's of all parameters are
$\hat{\beta}=0.5857$, $\hat{\alpha}_1=2.6297$, $\hat{\alpha}_2=2.3916$, and their standard errors are
$s.e.(\hat{\beta})=0.1766$,   $s.e.(\hat{\alpha}_1)=3.1333,$ $s.e.(\hat{\alpha}_2)=2.6609$.

\section*{Acknowledgements}
The author is grateful to the Editor-in-Chief and referees for their helpful comments and suggestions on improving the initial
version of this manuscript.

\bibliographystyle{apa}

\end{document}